\documentclass[12pt]{article}
\usepackage[utf8]{inputenc}
\usepackage{amsmath,amsthm,amssymb}
\usepackage{float}
\usepackage{tikz}
\usetikzlibrary{positioning}
\usepackage{graphicx}
\graphicspath{ {./} }
\usepackage{xcolor}
\usepackage{url}

\definecolor{darkspringgreen}{rgb}{0.09, 0.45, 0.27}

\usepackage{array,ragged2e}
\newcolumntype{P}[1]{>{\RaggedRight\arraybackslash}p{#1}}

\newtheorem{theorem}{Theorem}[section]

\newtheorem{lemma}[theorem]{Lemma}

\begin{document}

\title{NP-Hardness of a 2D, a 2.5D, and a 3D Puzzle Game}
\author{Matthew Ferland, Vikram Kher}
\maketitle
\begin{abstract}
In this paper, we give simple NP-hardness reductions for three popular video games. The first is Baba Is You, an award winning 2D block puzzle game with the key premise being the ability to rewrite the rules of the game. The second is Fez, a puzzle platformer whose main draw is the ability to swap between four different 2-dimensional views of the player's position. The final is Catherine, a 3-dimensional puzzle game where the player must climb a tower of rearrangable blocks.
\end{abstract}

\section{Introduction}
Games have been a major component of computer science research nearly since the inception of the field. Early on, research was focused entirely on algorithmic problems, notably chess \cite{shannon1950xxii}. In the 1970s, this trend shifted with the emergence of computational complexity research. Researchers began demonstrating the intractability of games which were typically combinatorial in nature \cite{schaefer1978complexity}. Despite video games (that would later shown to be intractible) becoming popular in the 1980s, it has only been within the past 20 years that there have been relevant intractability results. A few examples include showing the NP-hardness of generalized versions of Tetris \cite{demaine2003tetris}, the intractability of handful of NES and SNES games \cite{aloupis2015classic}, and the undecidability of a PC puzzle game \cite{demaine2020recursed}.

We will add to this recent trend by presenting three new hardness results for some recent block-based puzzle games, where we believe each proof is simple enough to be shown in an undergraduate algorithms course. For one of these games, the classification is tight.

The first of these games is Baba Is You, which released for PC and Nintendo Switch on May 13, 2019 to critical acclaim and commercial success. The game either won or was nominated for a dozen awards, including winning two awards at the Independent Games Festival before its release \cite{wik}. The developers were invited to talk about their experiences at the 2020 Game Developer Conference.\cite{BabaGDCTalk}

The basic premise of the game is that a player token needs to reach a goal token. However, the player may move around semantic blocks that determine various rules and interactions between objects. Regardless of the modifcations to the ruleset, the goal is always to reach some token.

The second of these games is Fez, a puzzle platformer released on April 13th, 2012 to universal acclaim \footnote{\url{https://www.metacritic.com/game/pc/fez/critic-reviews}}. Like Baba Is You, it won or was nominated for several awards. By 2014, it had sold over a million copies \cite{FezSales} and was cited as an inspiration for several later video games \cite{FezInsp1}\cite{FezInsp2}\cite{FezInsp3}.

Fez's novel feature is that it is a 2-dimensional while also acting 3-dimensional. The player can change the perspective to any of the four orthogonal perspectives, and as such, each object has four faces that can be interacted with. The game's mechanics have various ways of exploiting this feature.

The final game we explore is Catherine, a puzzle game released in 2011 (the exact date varies by region). Catherine, like the previous two, is a game that both won and was nominated for several awards \cite{wikCath}. The game has sold over a million copies, and in 2019 received a remake. The game intertwines two mechanically separate sections. The first is a dating drama, while the second, which this paper focuses on, is a 3-dimensional puzzle game that takes place in dreams of the main character.

This puzzle game requires the player to ascend a tower of various blocks, which can be rearranged by the player in order to facilitate their climb. The ``selling point'' feature of these segments is the game's unique physics. Instead of falling being obstructed by having a block below, it is instead only requires that a block below touches via a single one dimensional ``edge."

\section{Baba Is You}
\subsection{Rules}
The core game mechanics of Baba Is You are simple to understand. There is a token the player controls on a grid. A player can move this token up, right, left, or down. These are the only possible inputs. The player's goal is to move the token to the same square as a goal token.

What makes the mechanics unique is that there are additional rules other than the ones above which are subject to player changes. All additional game rules are determined by "$x$ IS $y$" sentences, where $x$ is the identifier for one or several tokens, and $y$ is either another token or a modifier to the functionality of the token. The sentences can possibly contain one or more "AND"s. For example, we may have "FLAG AND CRAB IS WIN AND FLOAT" which means that both the flag and crab tokens have the "WIN" and "FLOAT" properties added to them.

These rules can be interacted with as they each exist in the form of blocks (the tokens, modifiers, "IS," and "AND" are all tiles). These blocks may be moved by the player token. If the player token attempts to move onto a tile that has a rule block, the rule block will be pushed to the adjacent spot in the direction the player token moved. Players can disable a rule by pushing a block to not be adjacent to the others, and then create a new rule by pushing another block in its place.

While there are several modifiers the game can use, we will only describe the ones relevant for our proof. First, "YOU" is a modifier which denotes the token(s) the player controls. Whatever is identified by this will move in the direction specified when the player chooses a direction, and will qualify for reaching the goal token. Second, "WIN" denotes the token(s) that the player needs to reach with the "YOU" token(s) to win the level. Next, "STOP" identifies tokens in which other tokens cannot pass through. If a token attempts to move or be moved to a tile that has rule "STOP," then the token will simply not be moved. Finally, "DEFEAT" indicates tokens that will destroy a player token if they occupy the same tile space.

For example, let's look at figure \ref{fig:example}. Here, we have rules "BABA IS YOU," meaning the player controls the white token, "FLAG IS WIN," meaning the player's goal is to have the white token touch the yellow flag token, "WALL IS STOP," meaning that the white token can't move through the grey wall tiles, and "SKULL IS DEFEAT," meaning our token will be removed upon touching the red skull, resulting in a loss. Clearly, our token is unable to get to the flag currently, since both the skull and the walls are stopping the token from being able to reach the flag. However, the player can move the rule blocks using the white token to fix this. By pushing "IS" down one tile, the "SKULL IS DEFEAT" rule will be disabled, allowing the white token to safely reach the flag token and solve the puzzle.

\begin{figure}[h!]
\begin{center}
\includegraphics[scale=.22]{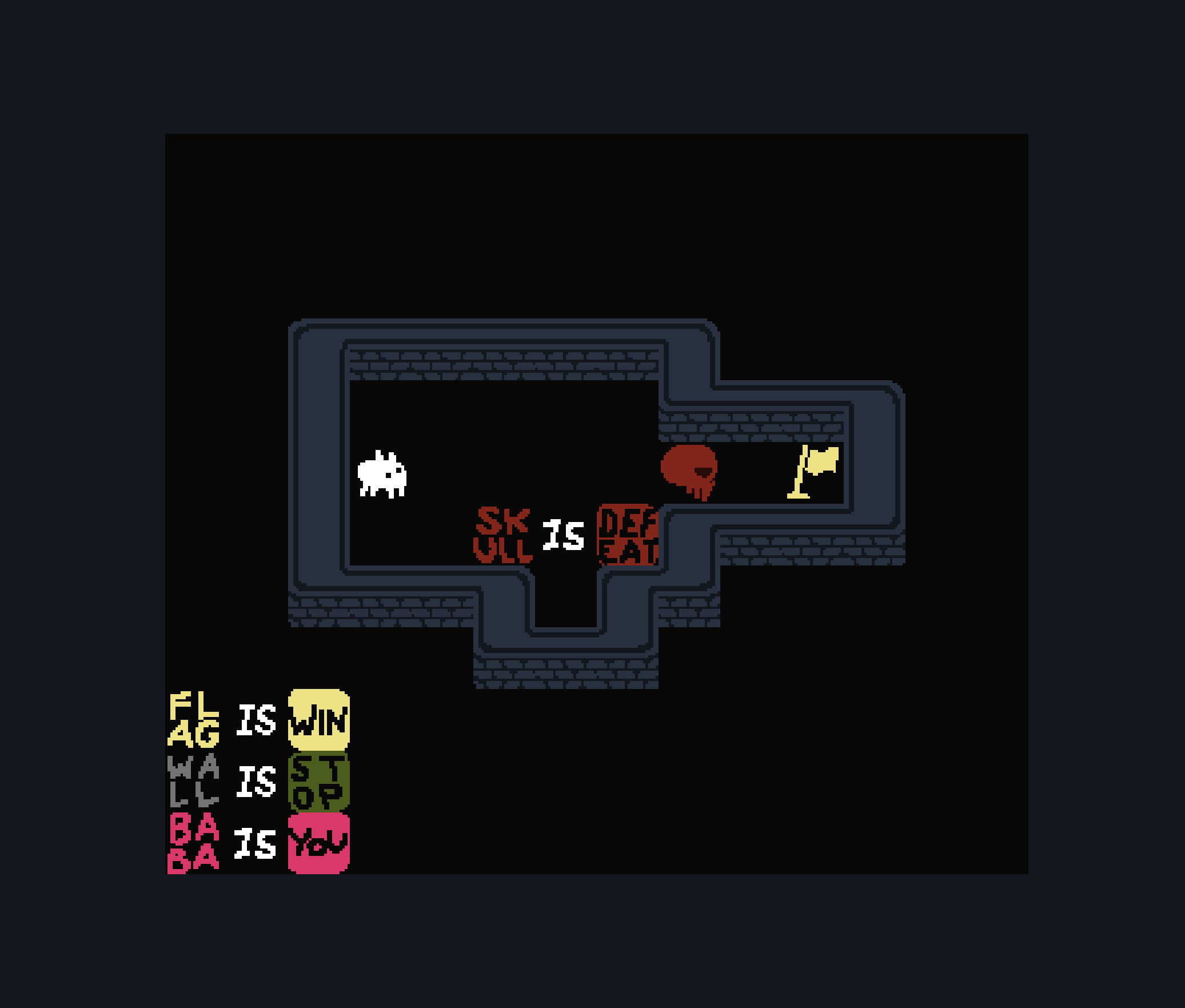}
\caption{Sample Baba Is You Puzzle}
\label{fig:example}
\end{center}
\end{figure}

\subsection{NP-Hardness}
We focus on showing the ruleset to be NP-hard. To this end, we will focus on being able to instantiate maps corresponding to any NP-hard problem. Here, we will reduce from 3-SAT.
\subsection{Reduction}
Our construction (as seen in figure \ref{fig:babareduction}) of the map begins with defining the following rules in the bottom left of the map: 'BABA IS YOU', 'FLAG IS WIN', and 'WALL IS STOP'. For each variable $x_i$ in the CNF, we create two unique tokens. We will refer to one of these as $y_i$ and the other as $\bar{y}_i$.

We represent each clause of 3-SAT in the map with a clause gadget (as seen in figure \ref{fig:babaclause}). Each clause gadget consists of three 1-block wide pathways that the player may choose to walk through. At the entrance of a given pathway, there exists $y_i$ or $\bar{y}_i$ token that corresponds to $x_i$ or $\bar{x_i}$ respectively in the clause.

After passing through token $y_i$, there is a rule block corresponding with the $\bar{y}_i$ token blocking the way (or alternatively $y_i$ after passing through $\bar{y}_i$). Then, there is an open space followed by "IS DEFEAT." Below and to the right of this construction is a wall. We connect the $n$ clause gadgets together sequentially and insert a flag at the end of the final clause gadget.

\begin{figure}[H]
\begin{center}
\includegraphics[scale=.25]{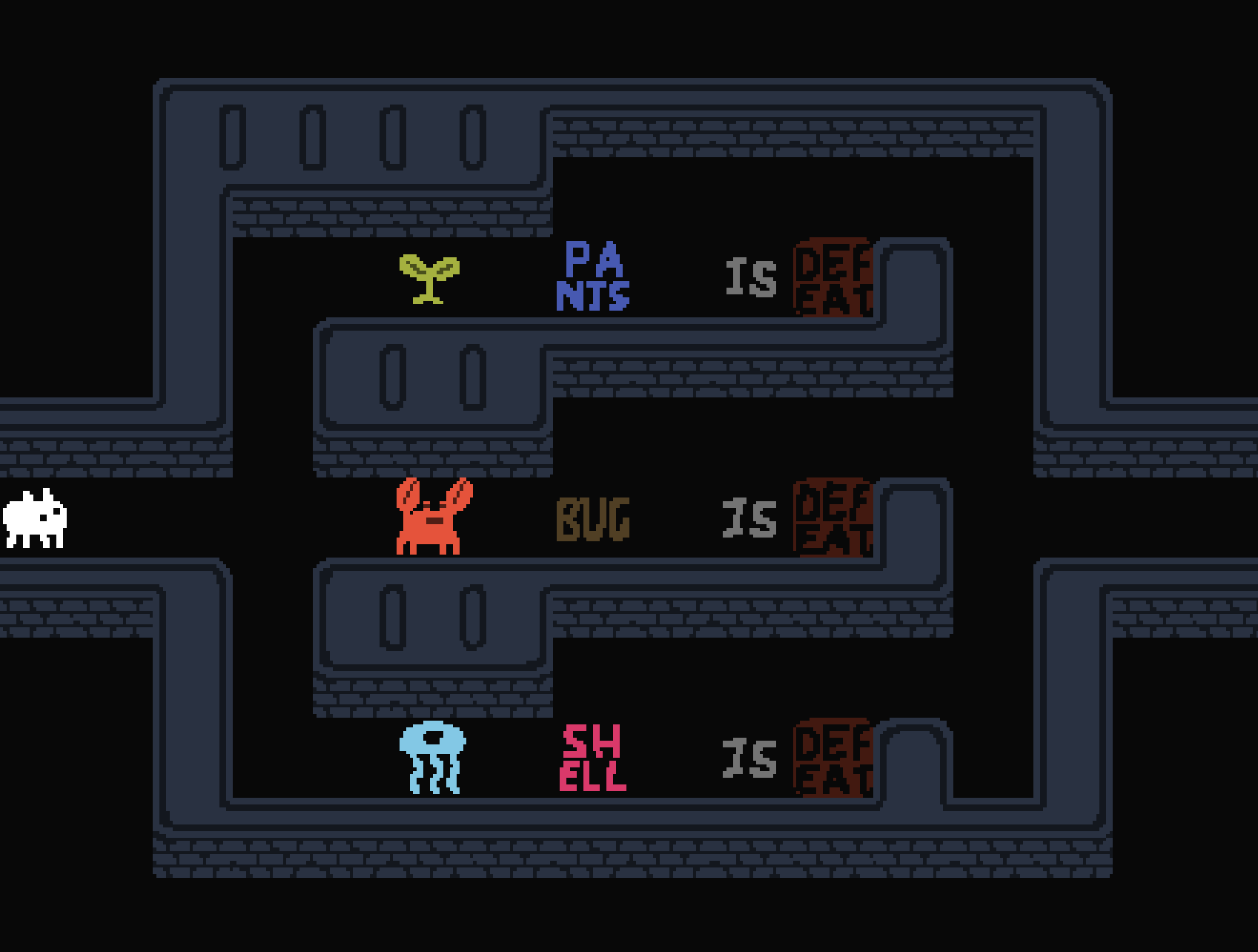}
\caption{Clause Gadget. Here, "PANTS" is the negation of the plant tokens, "BUG" is the negation of crab tokens, and "SHELL" is the negation of jellyfish tokens.}
\label{fig:babaclause}
\end{center}
\end{figure}

\begin{figure}[H]
\begin{center}
\includegraphics[scale=.12]{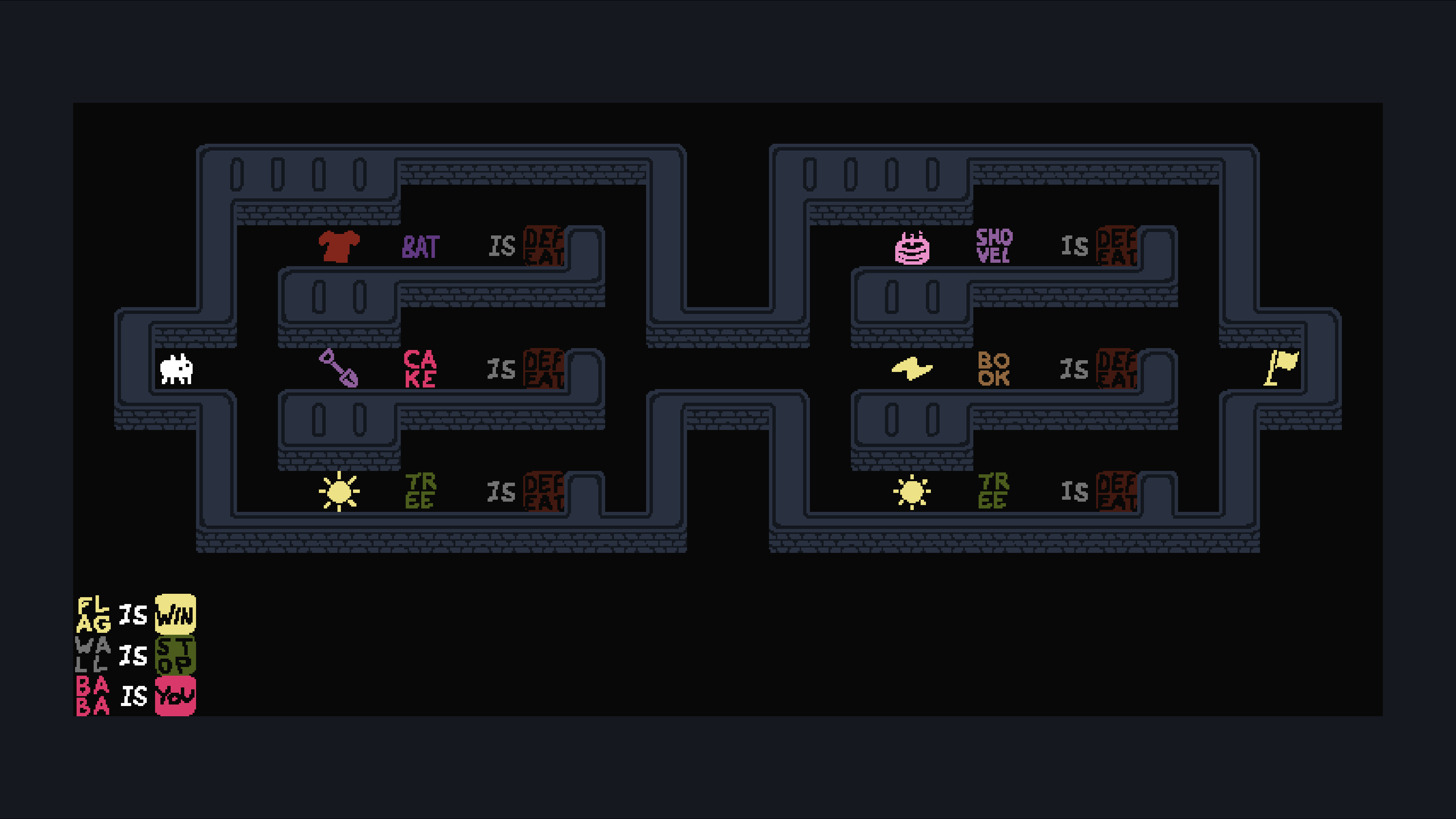}
\caption{Representation of $(x_1 \vee x_2 \vee x_3) \wedge (\overline{x}_2 \vee x_4 \vee x_3)$. Here, $x_1$ corresponds to shirt token, $x_2$ corresponds to the shovel token, $x_3$ corresponds to the sun token, $\overline{x}_2$ corresponds to the cake token, and $x_4$ corresponds to the lightning bolt.}
\label{fig:babareduction}
\end{center}
\end{figure}

\subsection{Correctness of the Reduction}
It now remains to verify that our embedding of 3-SAT is valid; namely that if our embedded 3-SAT formula is satisfiable then a path exists for the player to win the map.

If the corresponding 3-SAT is satisfiable, then for each clause there exists variable assignments that make it evaluate to true. If the player simply traverses each clause choosing a path corresponding to a variable assignment ($y_i$ if $x_i$ is true, and otherwise $\bar{y}_i$), they will eventually reach the flag since they will never enter a path where the token has the property DEFEAT, which only is active if one traversed a corresponding negation to the variable, which will never be assigned in a 3-SAT solution.

Conversely, we also wish to now verify that if a player can win the map, the embedded 3-SAT problem is satisfiable. If a path exists for a player to win the map then necessarily they must be able to traverse each clause. In order to traverse a clause, one needs to ensure that there is a path not corresponding to a negation of a previously selected variable. And since each clause must be traversed, there will exist a satisfying assignment of variables that makes each clause evaluate to true. Thus, all clauses are satisfied with the variable assignment from this path.

\begin{figure}[h!]
\begin{center}
\includegraphics[scale=.3]{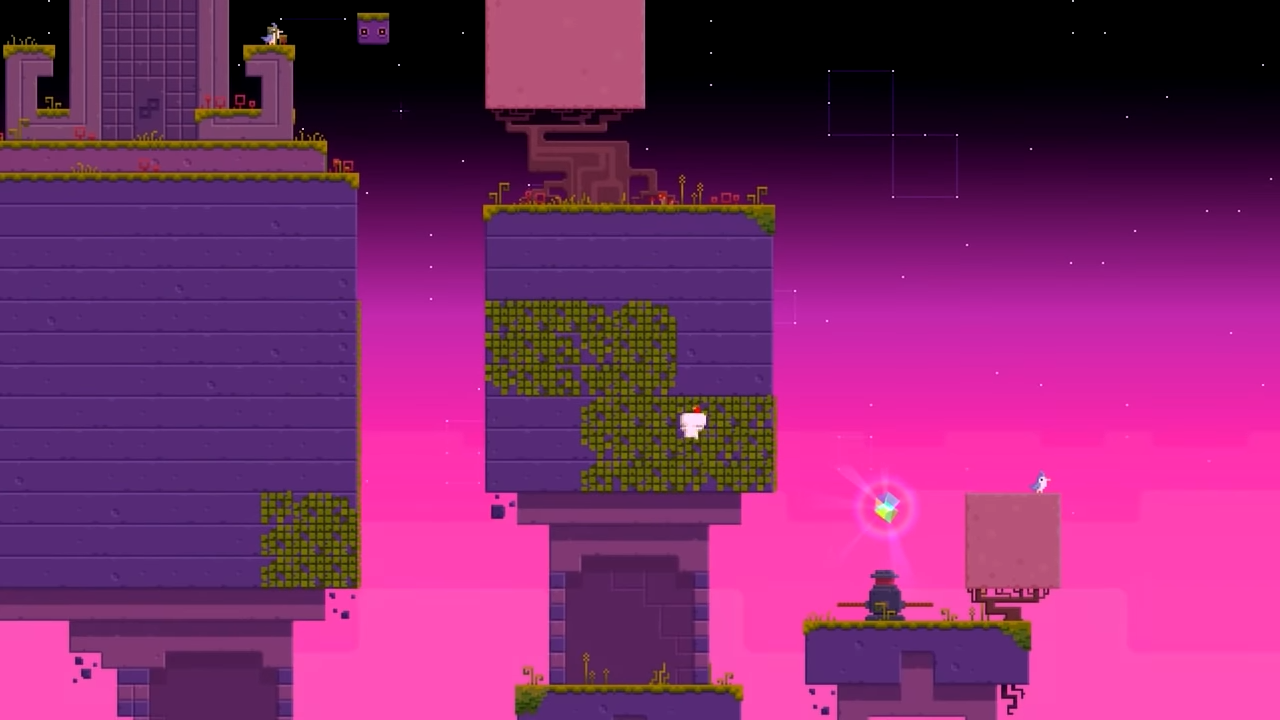}
\caption{Fez Sample Map}
\label{fig:fezSample}
\end{center}
\end{figure}

\begin{figure}[h!]
\begin{center}
\includegraphics[scale=.6]{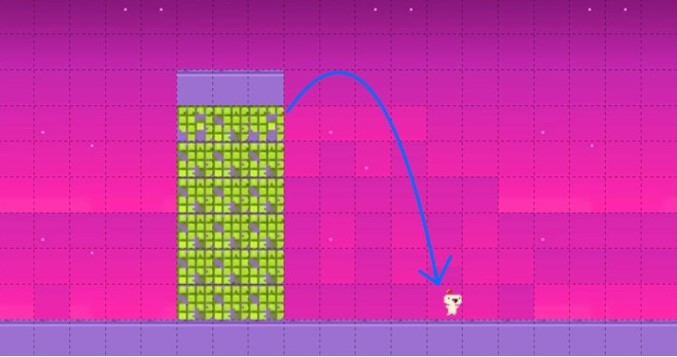}
\caption{Jumping from a vine}
\label{fig:fezjump}
\end{center}
\end{figure}

\section{Fez}
Fez is a platformer game with a variety of mechanics. For this reduction, we will use only a small handful of them which we will introduce now.

Fez's most notable mechanic is the shifting of perspective. One may think of the game world in 3 dimensions with the the player seeing through a camera that is always located on one face of a cube, excluding the top or bottom faces. The player may switch the perspective to a different face of the cube at any time. These switches rotate elements of the game world. Thus, it is a 3-dimensional game world that is interacted with on a 2-dimensional basis. Figure \ref{fig:fezSample} shows an example of a map in Fez. As with Baba is you, we will only introduce the minimum amount of information in order to construct the proof.

The player may make the character, Gomez, move left, right, or jump so long as there is a platform below that they are standing on. If there is nothing below, then they will fall. If Gomez is touching a vine (represented by a green texture on a wall), they may move around anywhere on the vine, even without a platform below, and can jump off of the vine at any time. These vines are attached to a face of a block.

When Gomez jumps off of a vine, he can travel a horizontal distance related to his current height compared to the ground. If Gomez is less than or equal to 4 blocks off the ground and he jumps, he can travel a maximum horizontal distance of 4 blocks. If Gomez is at a height greater than 4 blocks, then he can jump a horizontal distance of 4 blocks plus some additional distance that grows at a sub-linear rate compared to his height. For example, in figure \ref{fig:fezjump} we see Gomez's jump arc from a height of 6 and we see him land 5 blocks away from his starting position. Note, if Gomez jumped from a height of 2, he would land a maximum of 4 blocks away from his starting position.

The goal for the player is to pick up pieces of the Hexahedron, which are golden blocks. Finally, there are levers which can rotate rows of blocks 90 degrees, bringing one of the other faces to the current one. It is primarily the combination of vines and these rotations which allow us to construct NP-hard puzzles in Fez.

\subsection{Reduction}
We aim to embed a 3-SAT problem into a constructed in-game tower such that the player is able to horizontally traverse the tower from right to left and thus solve the puzzle if and only if the embedded 3-SAT is satisfiable. This 3-SAT problem has $n$ variables and $m$ clauses.

\begin{figure}[H]
\begin{center}
\includegraphics[scale=1.1]{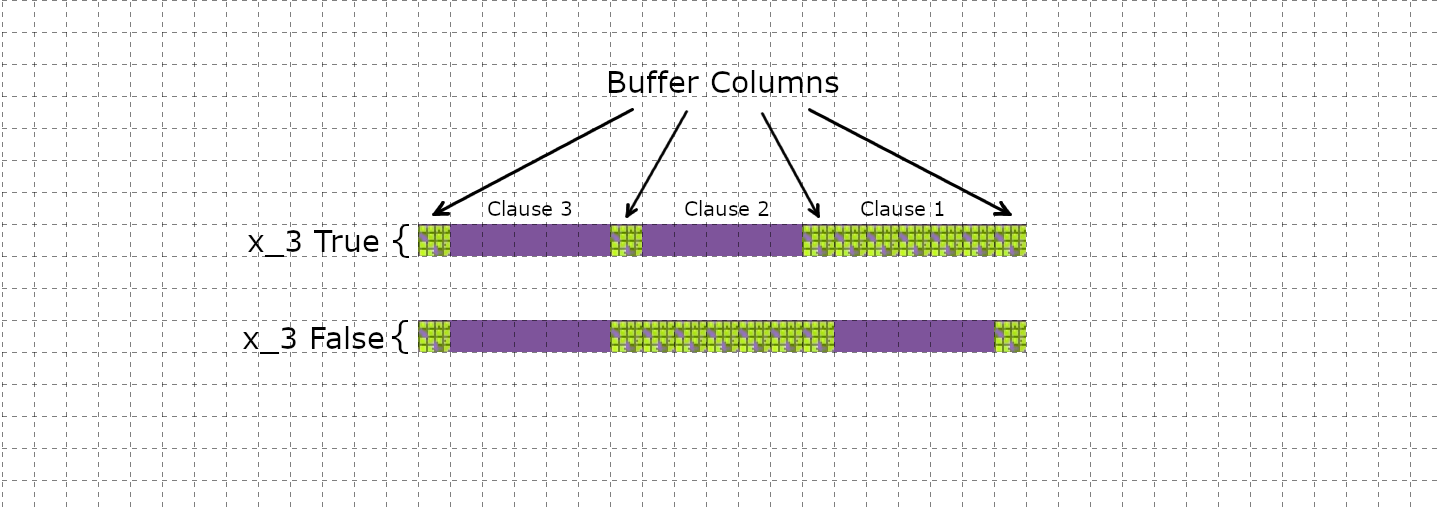}
\caption{Ring Gadget for $x_3$ from figure \ref{fig:fezreduction}}
\label{fig:fezColumns}
\end{center}
\end{figure}

\begin{figure}[h!]
\begin{center}
\includegraphics[scale=1.2]{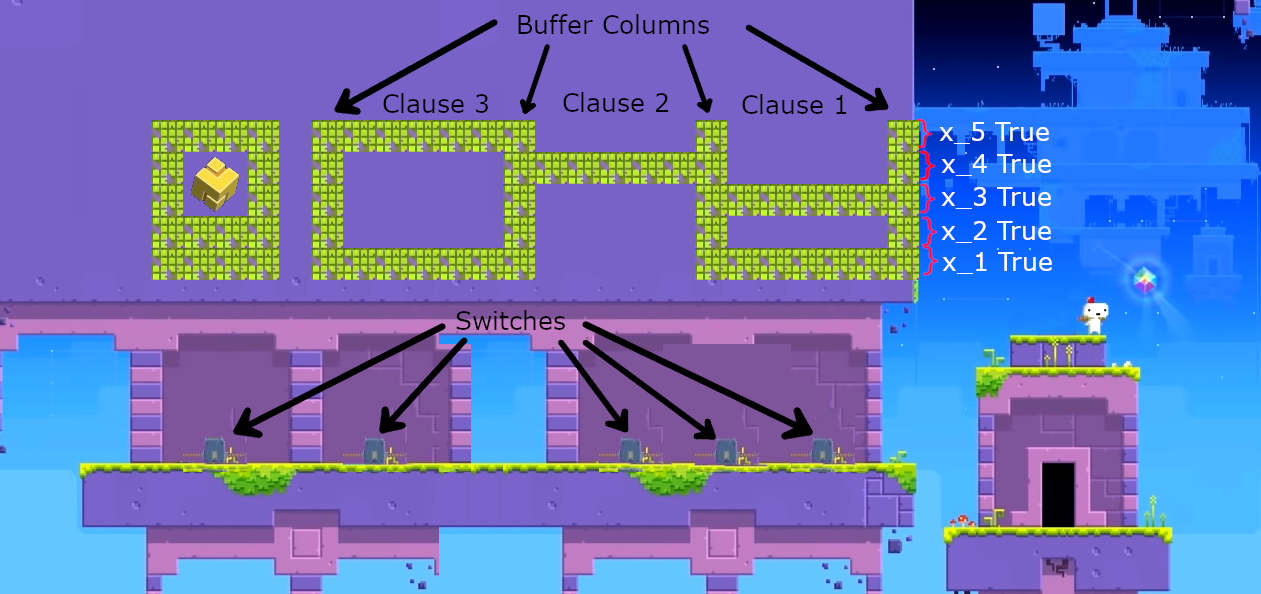}
\caption{Fez Construction: $(x_1 \lor \neg x_2 \lor x_3) \wedge (\neg x_3 \lor x_4 \lor \neg x_5) \wedge (x_1 \lor \neg x_4 \lor x_5)$. In this figure, all variables are set to true (in other words, they have their front face towards the camera). The bottom row is $x_1$, the second to bottom row is $x_2$, and so on, until the top row is $x_5$. The first clause is represented by the right most 5 columns, the second clause by the center five columns, and the last clause by the five columns left of that.}
\label{fig:fezreduction}
\end{center}
\end{figure}

We begin by defining the tower which consists of $n$ ring gadgets stacked vertically. Each ring corresponds to a variable in 3-SAT and has $(n+1)\cdot m$ columns (as seen in figure \ref{fig:fezColumns}). Each ring also has 4 faces: front, back, left, and right. We leave the left and right faces unused. We additionally connect each ring to a unique turnstile mechanism (a switch) that the player can manipulate to rotate each ring between these 4 faces.

We further divide each ring into $m$ segments each with $n+1$ columns. One of these columns serves as a buffer, while the other $n$ columns exist to enforce the clauses. For the buffer column, we simply let each row of its front and back faces contain vines. The idea is that these vines exist on each buffer column to allow Gomez to choose which variable he will use to traverse the next clause. For the $n$ "main" columns, for row $n_i$ and segment $m_j$, all of the columns in row $i$ have a vine on the front face if and only if variable $i$ appears unnegated in clause $j$, and have a vine on the back face if and only if variable $i$ appears negated in clause $j$. We leave the left and right faces of every segment without vines. The intuition is that the front of the ring corresponds to the variable being true, while the back represents the variable being false. For each clause an assignment satisfies, there will be a row of vines connecting the buffer columns that surround the clause.

We employ a large gap between buffer columns to prevent Gomez from being able to jump between buffer columns directly without the clause being satisfied. If $n < 5$, we will enforce that there are always at least $5$ main columns for each clause to prevent this jumping. Once, $n > 5$ then having a gap of distance $n$ is sufficient to prevent jumping. An example of our construction is shown in figure \ref{fig:fezreduction}.

We claim that the player is able to horizontally traverse this tower from right to left iff the embedded 3-SAT is satisfiable.

\subsection{Correctness of Proof}
Suppose a satisfying assignment exists for our embedded 3-SAT formula, we prove the player must be able to horizontally traverse the constructed tower and thus solve the puzzle.

Since a satisfying assignment exists for our embedded 3-SAT instance, at least one variable in each clause evaluates to true. We argue that this means that the player can rotate some subset of the rings through the turnstile mechanism in such a way that every "main" column contains at least one segment covered with player-facing vines. In our construction, we only placed vines on the front face of non-negated literals and on the back face of negated-literals. Initially, all of the rings' front faces are player-facing. Thus, we initially interpret all literals as being set to true. By rotating a ring twice (180 degrees), the player can set the literal that corresponds to the ring to false. Thus, a literal evaluating to true in a clause results in its corresponding segment in the ring having player-facing vines. Since a satisfying assignment exists, there must exist some combination of ring rotations that results in every segment having at least one row of vines.

Since each segment's first column has all vines, and contains a row that with all vines (through the satisfying assignment), and links to the first column of the next segment which has all vines, the player can traverse to each subsequent segment, letting the player reach the end.

Now, suppose the the player is able to horizontally traverse the constructed tower from right to left, we argue that the 3-SAT has a satisfying assignment. Since the player is able to traverse the tower, there must exist at least one row of only vines for each segment, since the player is unable to jump horizontally $n$ blocks from a height of $n$. Thus, there exists some configuration of the rows that let us achieve this.

Then, by construction, we can apply that configuration as assignments for the variables, achieving a true literal in each clause. Consequently, the embedded 3-SAT instance is satisfiable.

\section{Catherine}
\subsection{Rules}
In Catherine, the player controls a character called Vincent. In each level, Vincent's goal is to climb a tower of blocks to reach its top.

He has a few basic controls. He may move to any adjacent block that is on his same vertical level as the one he is standing on. He may climb on top of an adjacent block that is exactly one space higher than the block he currently stands on. Similarly, he may step down to an adjacent block that is exactly one lower than his current block. If Vincent attempts to step down in a direction where there is no block that is precisely one level lower, he will hang from the edge of his current block. From here the player may either let go of the block to drop down, or continue hanging while being able to traverse to adjacent blocks at the same height. For example, in figure \ref{fig:cathhang}, Vincent is hanging from a block and can choose to either move to the left or right of him, or drop down to the block below him. Notably, Vincent does not have the ability to move up to a vertically higher block, and is stuck traversing horizontally at that level.

The main puzzle mechanic comes from Vincent's capacity to move blocks. Vincent may push a block in the direction opposite of where he stands. Similarly, he may pull a block in the direction he is currently standing. Vincent cannot pull a block if there is a block placed behind him, though he can pull a block if there is a gap behind him, since he will simply hang on the ledge.

If a block has at least one of its bottom edges touching the top edge of a block below it (as in, at least 2 vertices of a cube below it), then it will remain in place, regardless of how blocks below are moved. If there is no such edge, the block will fall until it encounters another block's edge, or until it reaches the bottom of the map, where it disappears.  For example, figure \ref{fig:cathedge} shows two red blocks that form a bridge over a gap. Each red block is supported by the blue block that is adjacent to it. If one of these blue blocks is removed, it's adjacent red block would fall by a height of 1.

In addition to standard blocks, towers can contain several special types of blocks. One special block we will use in our reduction is a \textit{type-2 cracked block}, which may only have up to two instances of Vincent and/or blocks passing over it before it disintegrates and disappears. Thus, we say that a \textit{type-2 cracked block} initially has a durability of 2 and disappears when its durability reaches zero. Note, there is a \textit{type-1 cracked block} but we do not use it in our reduction. So, from here on, we will simply use the term \textit{cracked block} to denote \textit{type-2 cracked blocks}. Additionally, we note that certain blocks cannot be moved. In the following figures, the only blocks that Vincent can move are those colored green. Grey, blue, and red blocks are not movable by him.

There is also effectively a timer for Catherine's levels, which we refer to as the death plane. The death plane is a invisible plane that slowly rises, incrementally destroying layers of the tower in the process. As the name implies, Vincent will fall to his death if the plane reaches the layer he is currently standing on. Initially, this plane is below the playable level, but it will eventually reach the base of the tower. At least on the first map, the death plane rises and destroys one layer of the tower approximately every 10 seconds, so for our reduction we will conservatively assume this occurs once every 9 seconds. Furthermore, we will assume (again, conservatively) that Vincent can take one action every second. This is an acceptable assumption, since we will show that the tower in our reduction cannot be traversed even with an unlimited amount of time if we are reducing from a false instance of 3-SAT. Meanwhile, if we are reducing from a positive instance of 3-SAT, we will always be able to solve the puzzle within the time constraint.

\begin{figure}[H]
\begin{center}
\includegraphics[scale=.2]{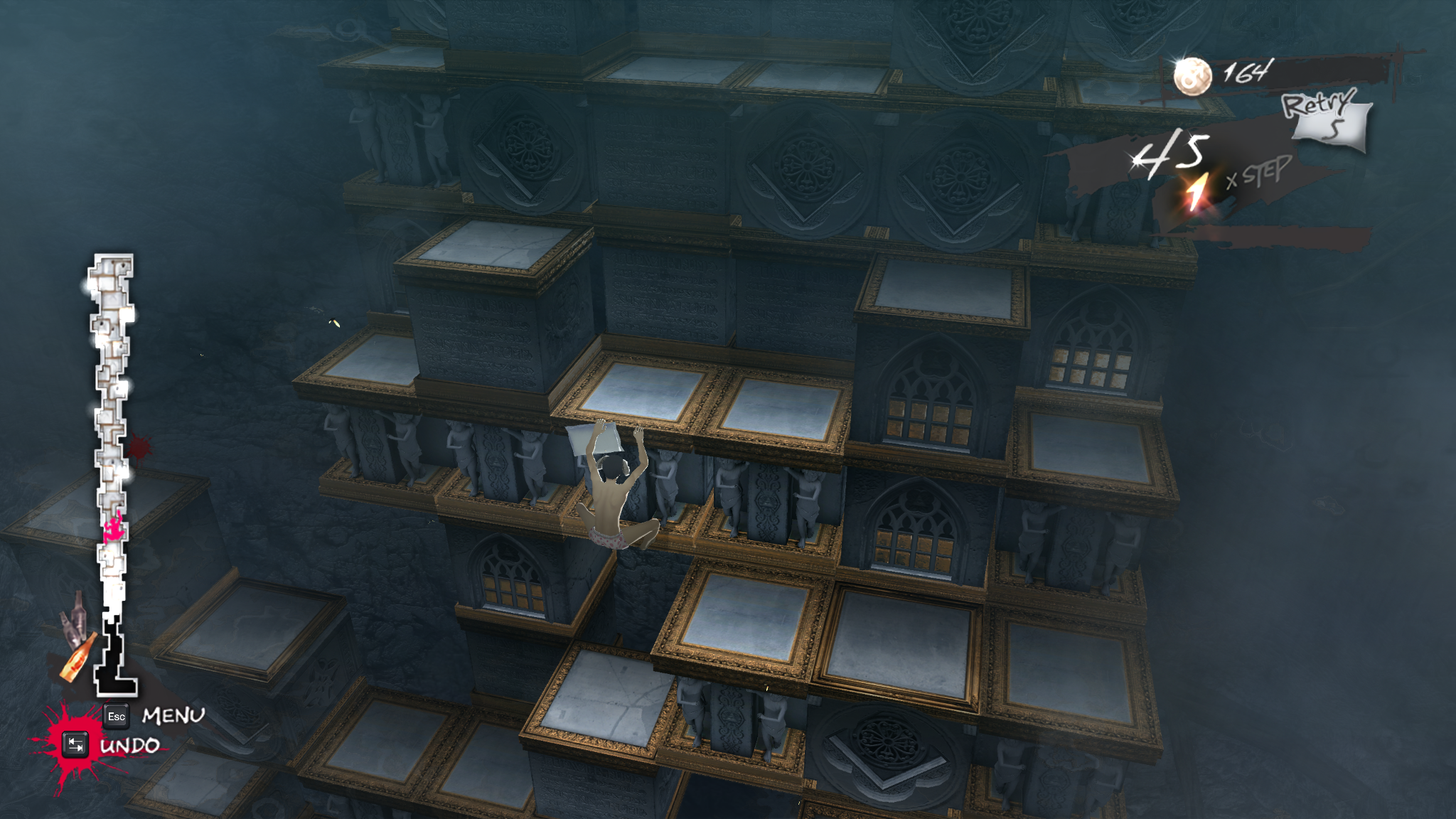}
\caption{Example of Vincent Hanging from a Block}
\label{fig:cathhang}
\end{center}
\end{figure}

\begin{figure}[H]
\begin{center}
\includegraphics[scale=.3]{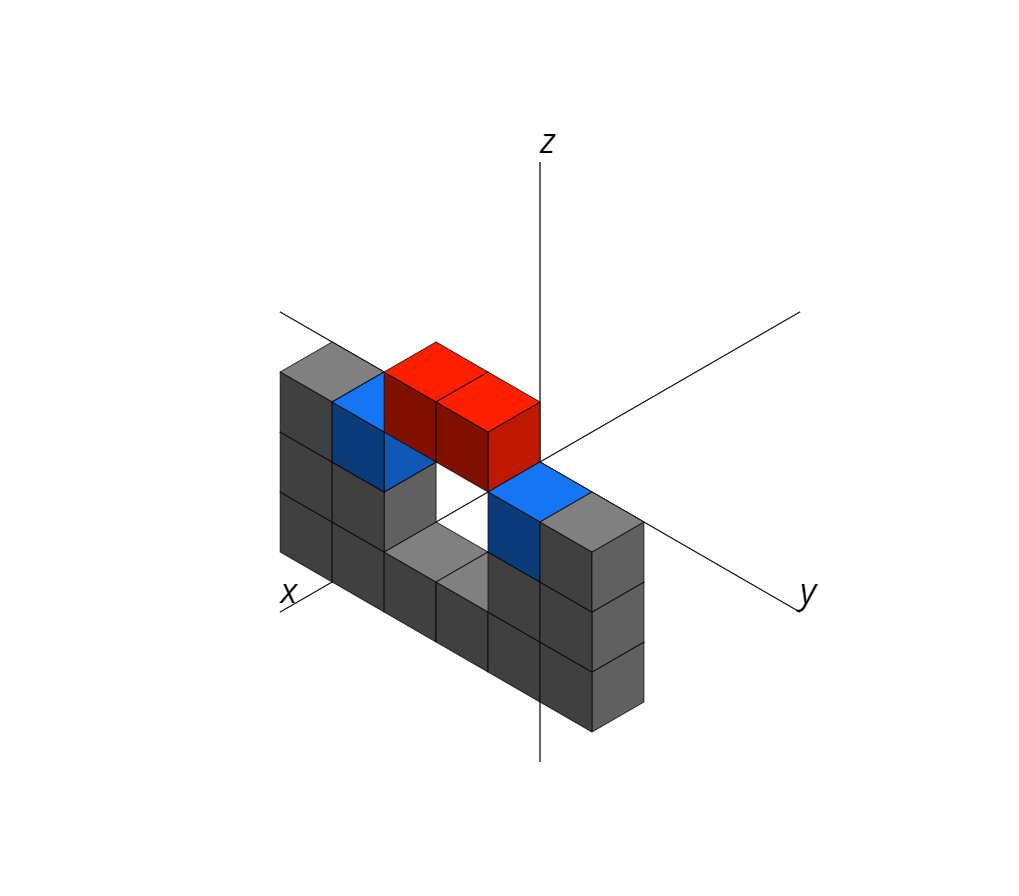}
\caption{Example of the Edge Property}
\label{fig:cathedge}
\end{center}
\end{figure}

\subsection{NP-Hardness}
We aim to show that the game is NP-Hard. We reduce from 3-SAT with $n$ variables and $m$ clauses to a constructed puzzle tower in the game.
\subsection{Reduction}
We begin by constructing a block tower of $O(n^2 + m)$ height, $O(n)$ width, and constant depth. The tower is divided into two sections, the staircase section and the main section. The staircase section consists of $2n^2 + 2m$ stacked spiral staircase gadgets (as seen in figure \ref{fig:cathstaircase}). By climbing these spiral staircases, Vincent will gain enough time to traverse the main section of the tower before the death plane destroys it.

Vincent can traverse the spiral staircase gadget, shown in figure \ref{fig:cathstaircase}, in a counter-clockwise fashion. Specifically, he may walk two blocks forward then ascend the block in front of him and finally rotate 90 degrees to the left. By repeating this process three times, Vincent will traverse a single spiral staircase gadget and gain three blocks of height, after only 8 seconds have passed.

The stacked spiral staircase gadgets allow Vincent to reach the tower's main section, which is of $O(m)$ height. The bottom floor of this section consists of the main platform gadget. This gadget, shown in figure \ref{fig:cathmain}, is a $6n$-block wide and $3$-block deep platform. The main platform gadget provides access to $n$ variable gadgets through two 1-block wide \textit{cracked block} pathways. The figure has only a portion of the main platform gadget, where we show a section of the platform (with gray blocks) and show the pathways to 2 different variable gadgets (with red blocks).

Each variable gadget (figure \ref{fig:cathvariable}) consists of a 5-block wide and 3-block deep platform with two towers of height $O(m)$ and width $2$ attached by 1-block to the variable platform. The red block pathways connect back to the main platform and each green block tower extends to a height of $O(m)$. Additionally, each tower is supported through the edge property by 1-block. Removal of this "connector" block will cause its respective tower to collapse. Intuitively, before either of these connector blocks are removed, the variable gadget is simultaneously set to both True and False. By removing one of the connector blocks, we effectively eliminate one of the two assignments. Since each variable gadget is connected to the main platform gadget by two 1-block wide \textit{cracked block} pathways, it can only be entered and left a total of 4 times; this includes being entered by Vincent, being left by Vincent, and taking a block from the variable gadget to the main platform gadget.

The main platform gadget also connects to the clause gadgets via the gap gadget. The clause gadgets can be reached from the right side of the main platform gadget by crossing a 3-block wide and $n-3$ height pit (as seen in figure \ref{fig:cathmaingap}). Note, we say that this gap has three slots (numbered from left to right); this terminology is used to indicate the three different places where blocks may fall into. The blue block is a visual placeholder for $n-4$ vertically stacked blocks. This gap cannot be traversed by Vincent unless he collects and deposits $n$ blocks into the gap. Upon crossing this pit, there is a spiral staircase gadget that connects to the first clause gadget.

 Each clause gadget (figure \ref{fig:cathclause}) consists of a $6n$-block wide and $3$-block deep platform with 1-block placed on the north-west corner of the platform. There also exists up to three adjacent towers that are separated from the platform by a 1-block wide gap. A tower is adjacent to a clause gadget if and only if its corresponding variable evaluates to True in the clause. Figure \ref{fig:cathclause} displays a clause gadget with two green block towers, meaning that two out of three of the clause's variables evaluate to True. The base of these two towers originate from their respectively variable gadgets (recall figure \ref{fig:cathvariable}). Intuitively, the extra block placed on each platform can be used to bridge the 1-block wide gap between the platform and a tower. Using this bridge, the player could extract a block from a tower and pull it onto the platform. Clause gadgets are connected to each other by spiral staircase gadgets that can be reached by crossing a 2-block wide and 2-block high pit that is attached to the right side of the gadget (as seen in figure \ref{fig:cathclausegap}). Vincent must push two blocks into this gap to make it crossable (see figure \ref{fig:cathedge} for an example). The exit door of the tower can be reached by traversing the last clause gadget. This completes our encoding of 3-SAT.

\begin{figure}[H]
\begin{center}
\includegraphics[scale=.4]{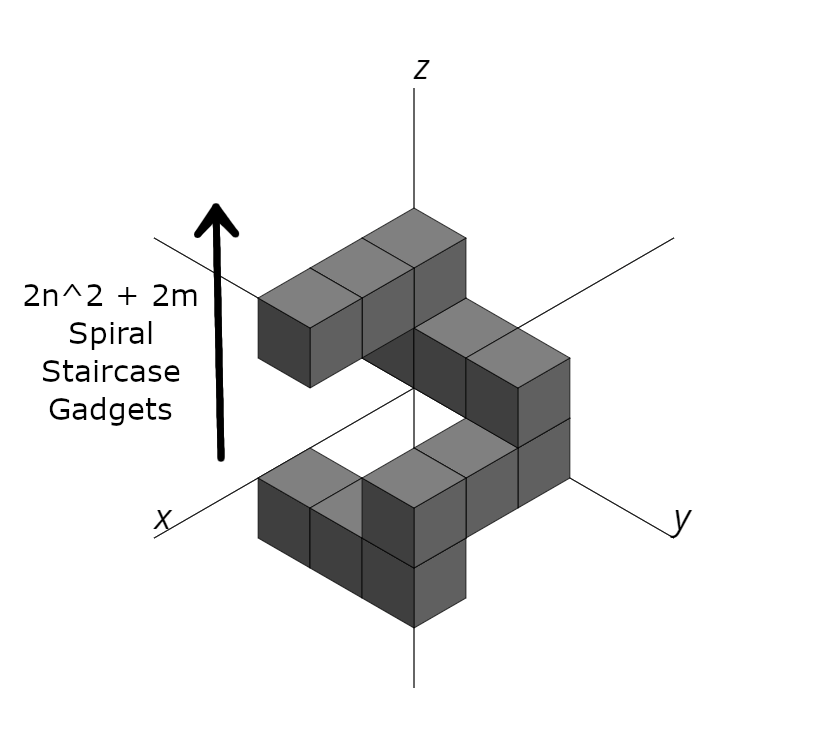}
\caption{Spiral Staircase Gadget}
\label{fig:cathstaircase}
\end{center}
\end{figure}

\begin{figure}[H]
\begin{center}
\includegraphics[scale=.4]{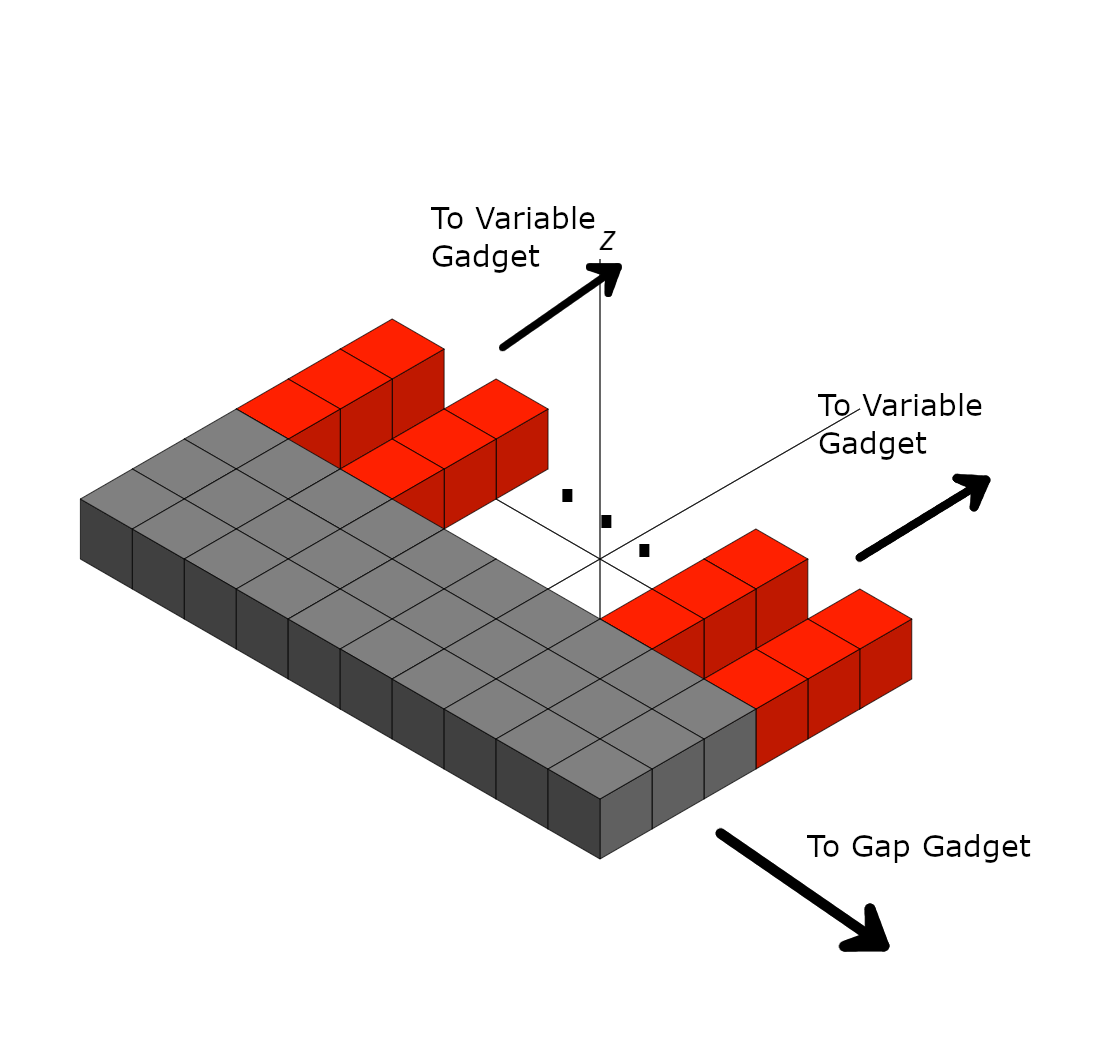}
\caption{Main Platform Gadget}
\label{fig:cathmain}
\end{center}
\end{figure}

\begin{figure}[H]
\begin{center}
\includegraphics[scale=.4]{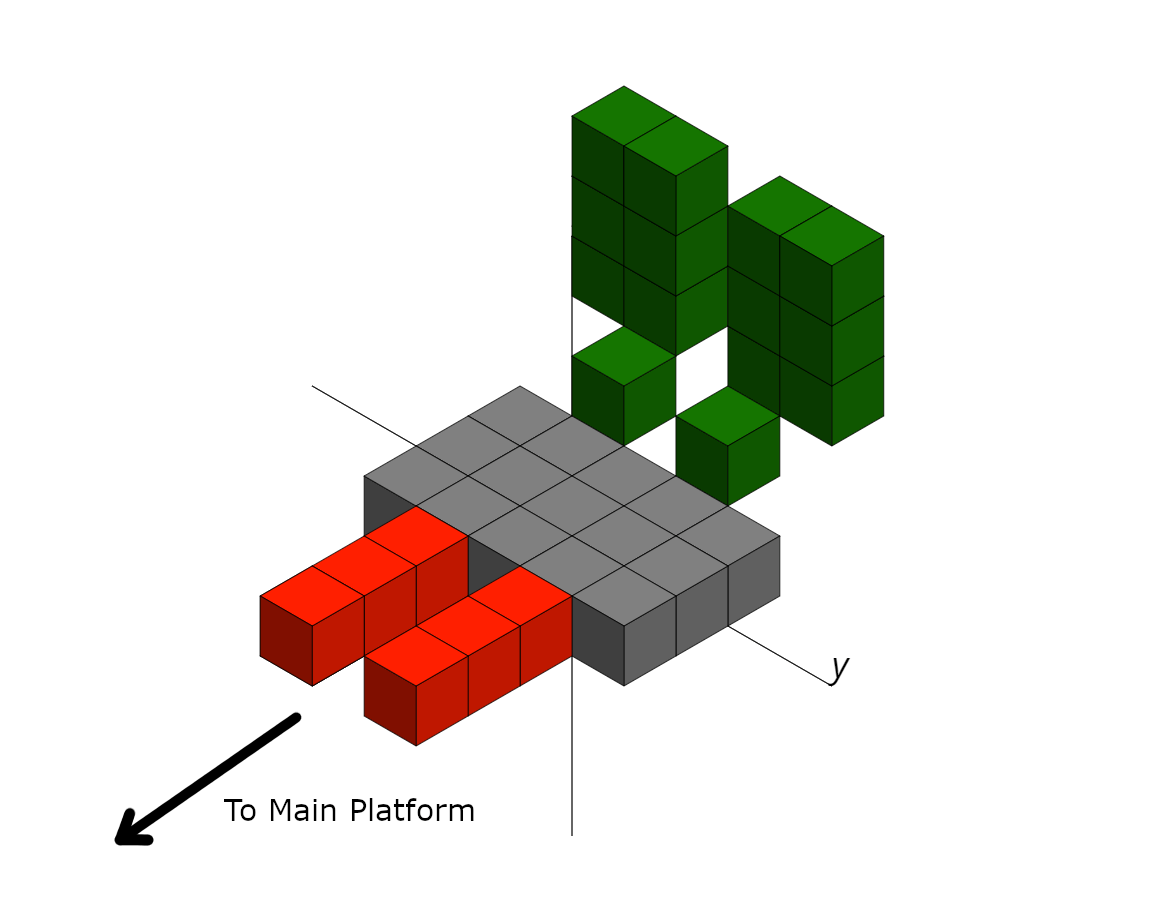}
\caption{Variable Gadget}
\label{fig:cathvariable}
\end{center}
\end{figure}

\begin{figure}[H]
\begin{center}
\includegraphics[scale=.4]{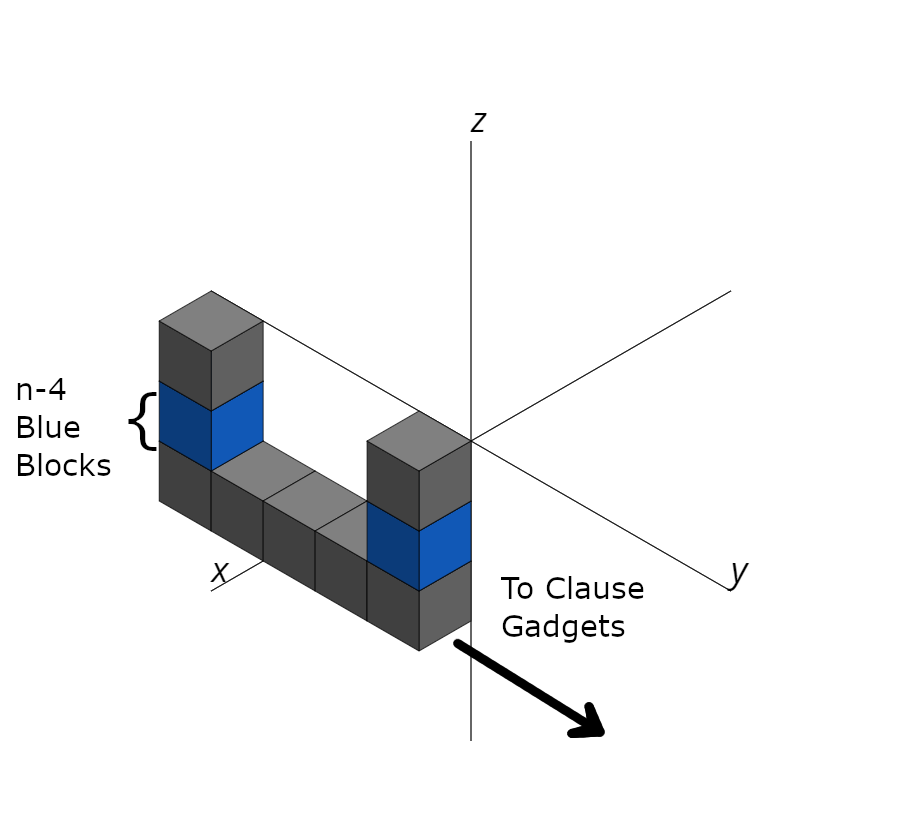}
\caption{Gap Gadget}
\label{fig:cathmaingap}
\end{center}
\end{figure}

\begin{figure}[H]
\begin{center}
\includegraphics[scale=.3]{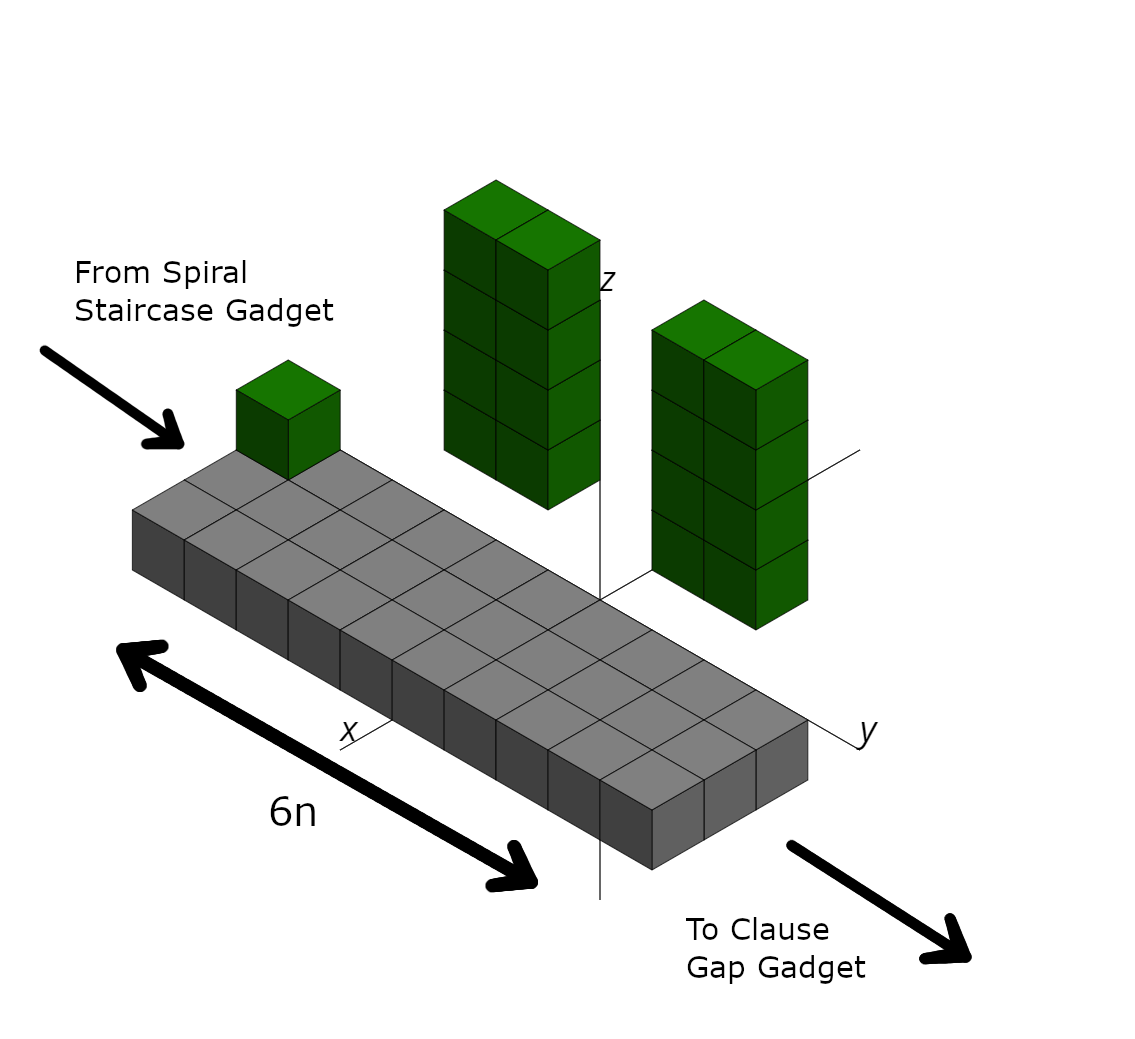}
\caption{Clause Gadget}
\label{fig:cathclause}
\end{center}
\end{figure}

\begin{figure}[H]
\begin{center}
\includegraphics[scale=.3]{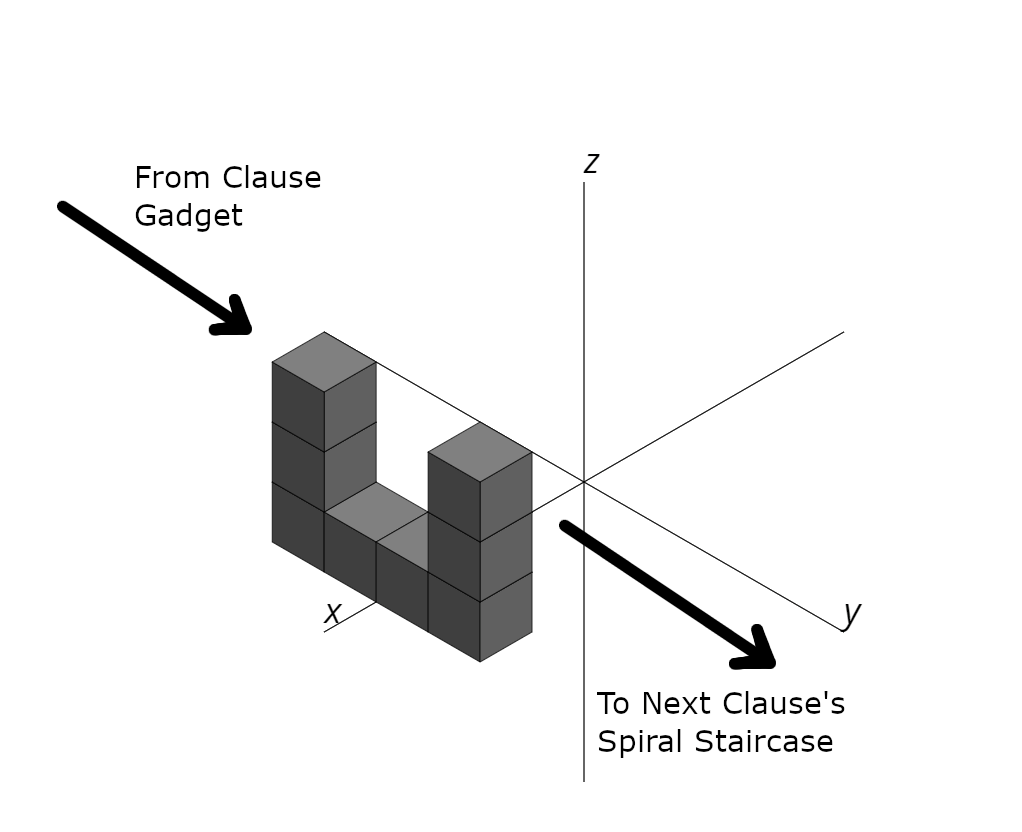}
\caption{Clause Gap Gadget}
\label{fig:cathclausegap}
\end{center}
\end{figure}

\begin{lemma}\label{lem1}
The gap gadget can be crossed if and only if $n$ blocks are pushed into it.
\end{lemma}
\begin{proof}
First, we will point out that one can cross the gadget with $n$ blocks. First, recall that the gap gadget is 3-blocks wide and consequently has three sequentially labeled slots where blocks may fall into. The first block pushed into the gap will hover over the gap's first slot due to the edge property. By pushing a second block into the gap, the first block will shift into the second slot of the gap and will no longer be supported by the edge property. As a result, the block will fall to the bottom of the second slot. The second block will now be hovering over the first slot (again supported by the edge property).

By conducting this procedure a total of $n-2$ times, the middle slot of the gap will completely filled with $n-3$ blocks and the first slot will have 1-block (the $n-2$th block) hovering over it due to the edge property. The player could now push two more blocks into the gap to make it crossable. When the $(n-1)$th block is pushed into the gap, it would shift the $n-2$ block from the first slot to the second slot. Since the second slot is completely filled, the block would not fall. If the player now pushes the $n$th block into the gap, it would shift the $(n-1)$th block to the second slot and the $n-2$ block to the third slot. The $n$th and $(n-2)$th blocks would not fall since they are both supported by the edge property. The $(n-1)$th block would not fall since it is in the second slot, which again is completely filled. Since every slot in the gap is covered, the player may now cross the gap. Figure \ref{fig:cathgapcorrect} shows a crossable gap gadget when we have $n=6$ variables.

On the other hand, the gap gadget cannot be crossed by pushing less than $n$ blocks into it. Since Catherine contains no jump mechanic, Vincent will not be able to cross the gap by just jumping over it. If Vincent were to fall into the gap without pushing any blocks into it, he would be stuck there since he cannot pull any blocks from the surrounding walls and has no ability to jump.

Now suppose Vincent pushes $k$ blocks into gap for some $k < n$. Vincent can only push blocks into the gap in same manner as described previously. There is no other way for Vincent to push blocks into the gap. As a result, there would be 1-block hovering over the first slot and $k-1$ blocks that have fallen into the second slot. Since the third slot is still not covered or filled in any way, the gap would not be crossable. An example of this scenario when we have $n=6$ variables is shown in figure \ref{fig:cathgapincorrect1}.

The player could try to fall into the gap after pushing $k$ blocks into it. If the player pulled the block hovering over the first slot away from the gap, they could drop into the first slot. Once there, they would not be able to get out of the gap. The only action they could perform is pushing the adjacent block in the second slot into the third slot. Due to the edge property, stacked blocks in the second slot would not fall. The player could now move into the second slot (see figure \ref{fig:cathgapincorrect2}). From there, they could either go back to the first slot or pull the block from the third slot back into the second spot. In both cases, the player is back in the first spot and is no closer to escaping the pit. Vincent does not have the ability to pull blocks from the walls of the gap.

If the player dropped into the second slot of the gap, then they would be on the top of $k-1$ stacked blocks. From there, they could only choose to either drop down the first or third slot of the gap. If they dropped down to the first spot, the player would be in the exact same situation as the previous case. If they dropped to the third spot, the player would only have the option to push the adjacent block in the second slot into the first slot (see figure \ref{fig:cathgapincorrect3}). This is symmetric to the case when the player is in the first slot and thus, they would not be able to cross the gap.

Consequently, the gap gadget is crossable by Vincent if and only if he pushes $n$ blocks into it.
\end{proof}

\begin{figure}[H]
\begin{center}
\includegraphics[scale=.3]{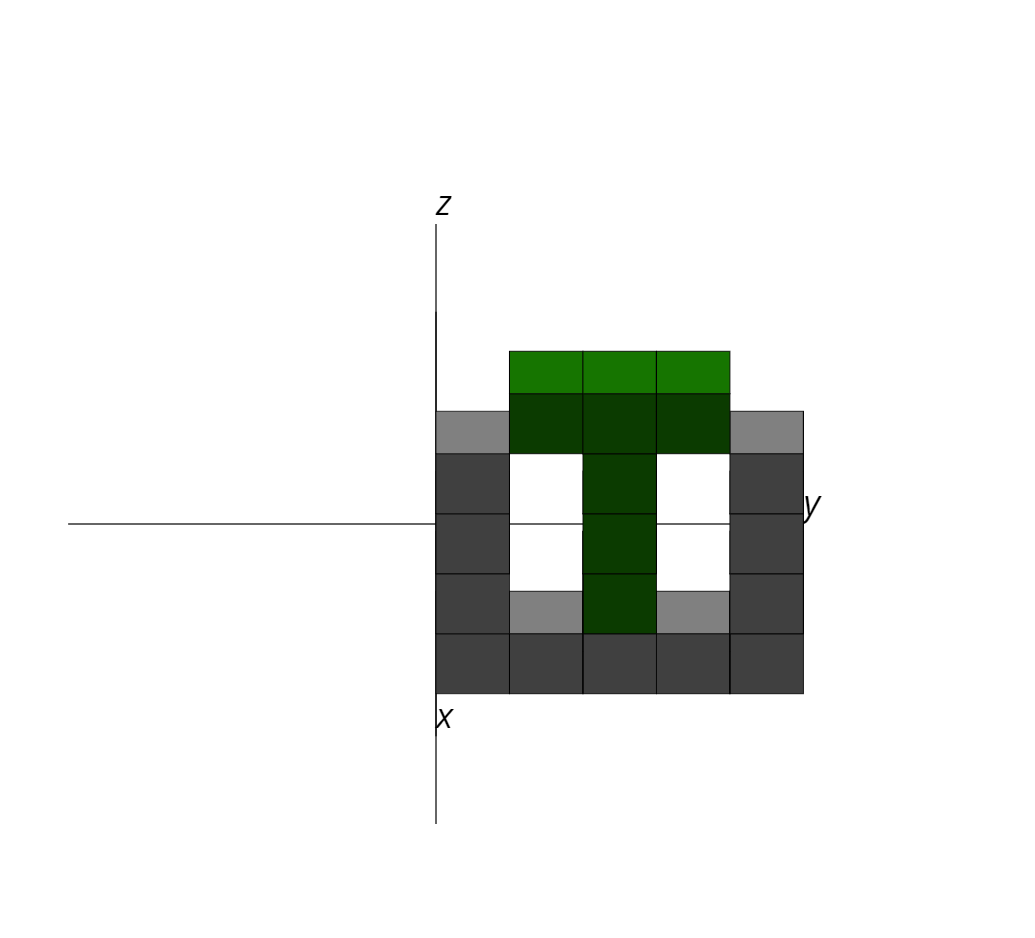}
\caption{A gap gadget for a 3-SAT instance with $n=6$ variables. By pushing in 6 green blocks, the gap is made crossable.}
\label{fig:cathgapcorrect}
\end{center}
\end{figure}

\begin{figure}[H]
\begin{center}
\includegraphics[scale=.3]{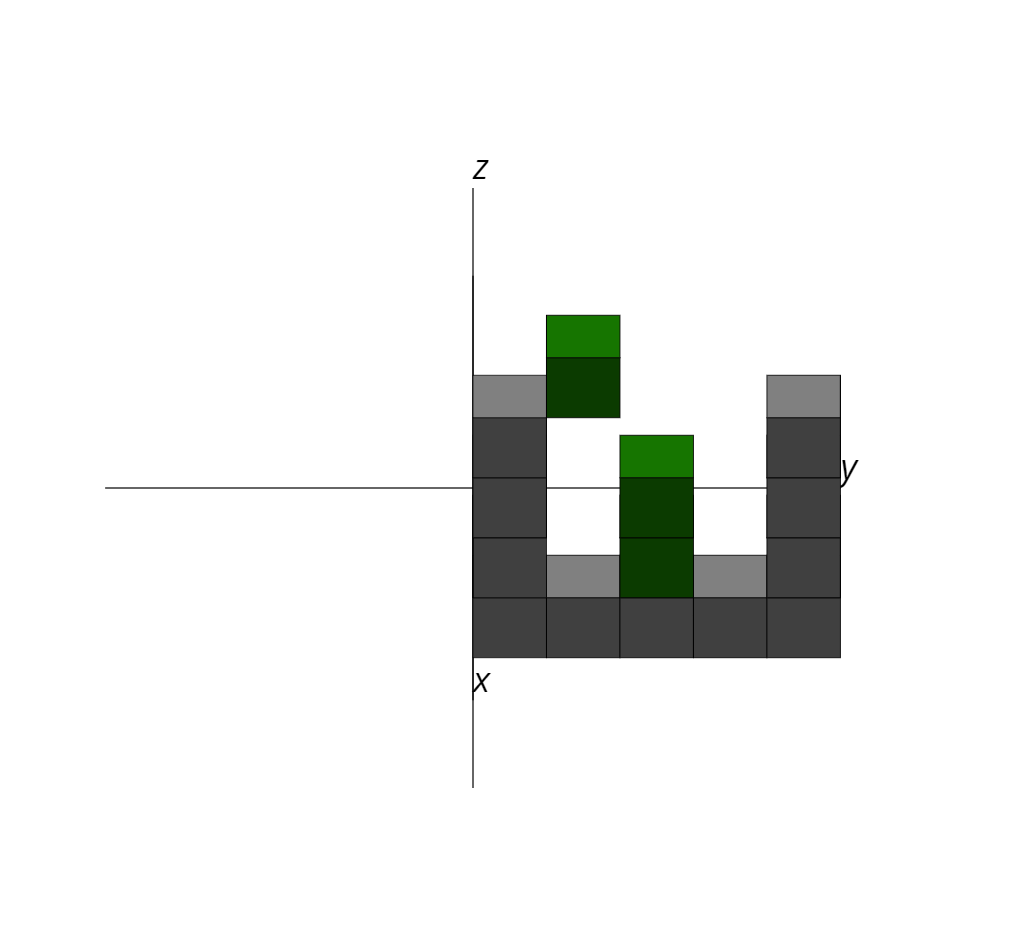}
\caption{A gap gadget for a 3-SAT instance with $n=6$ variables. The gap is not crossable.}
\label{fig:cathgapincorrect1}
\end{center}
\end{figure}

\begin{figure}[H]
\begin{center}
\includegraphics[scale=.3]{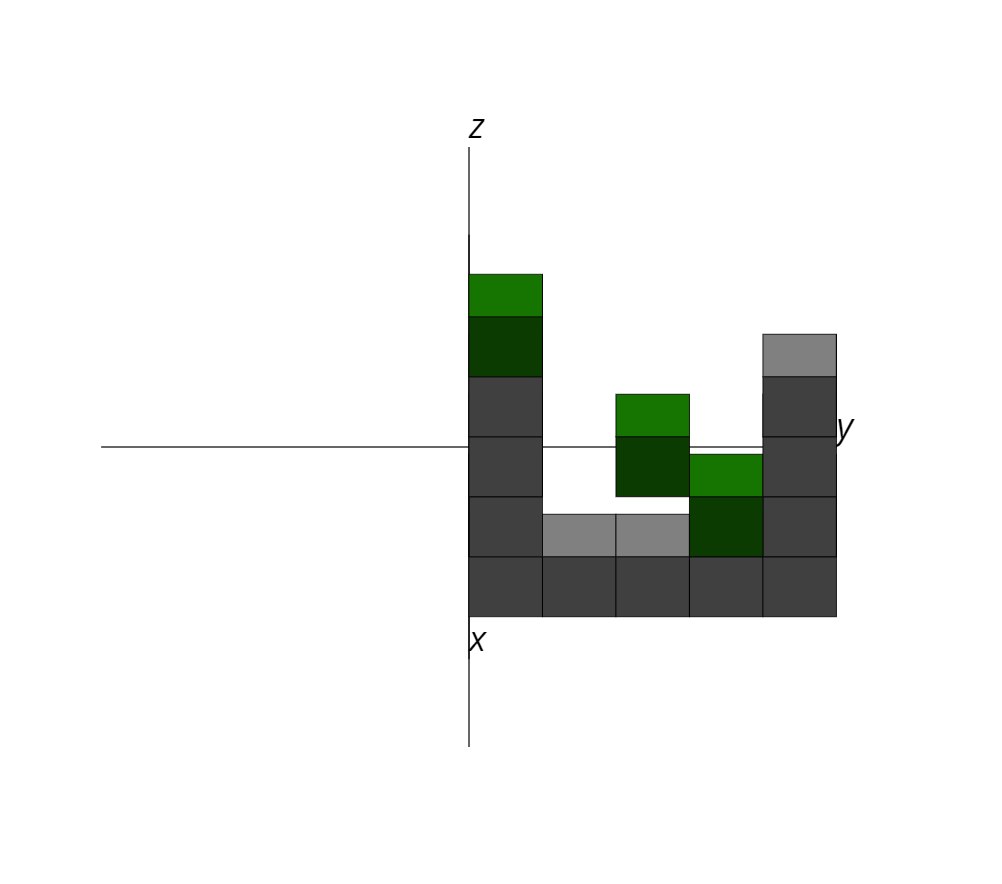}
\caption{A gap gadget for a 3-SAT instance with $n=6$ variables. The gap is not crossable.}
\label{fig:cathgapincorrect2}
\end{center}
\end{figure}

\begin{figure}[H]
\begin{center}
\includegraphics[scale=.3]{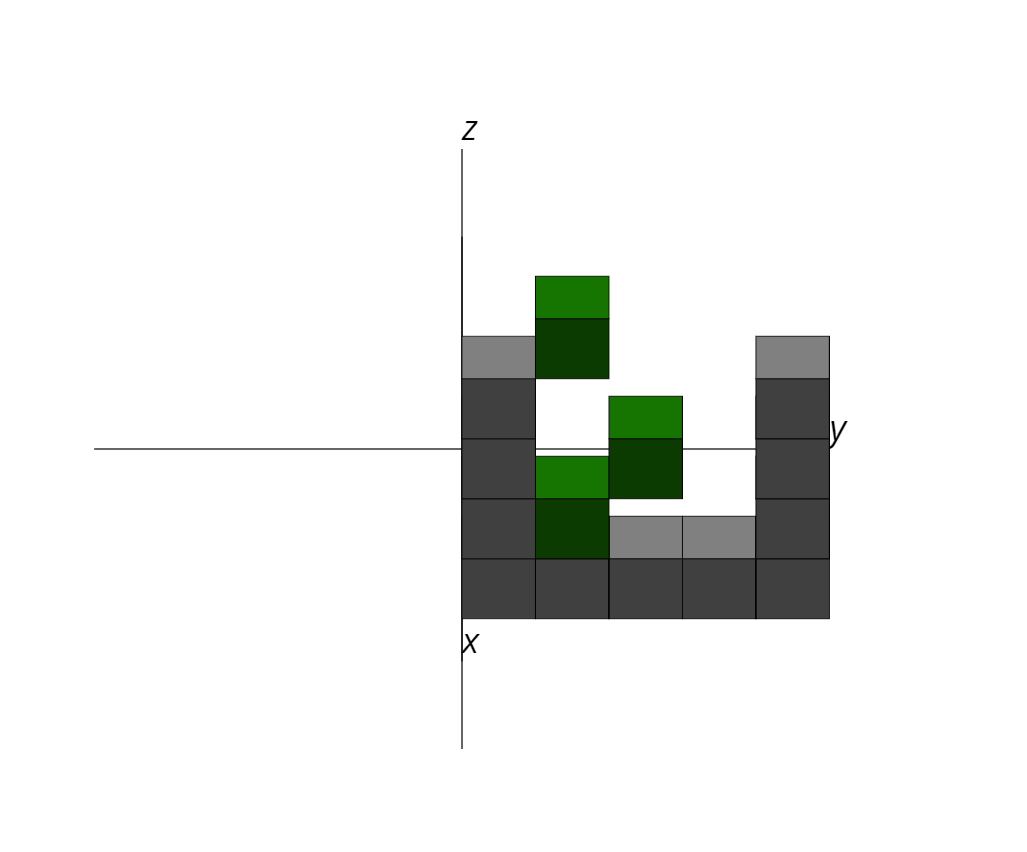}
\caption{A gap gadget for a 3-SAT instance with $n=6$ variables. The gap is not crossable}
\label{fig:cathgapincorrect3}
\end{center}
\end{figure}

\subsection{Correctness of the Reduction}
Suppose a satisfying assignment exists for our embedded 3-SAT formula, we prove that a path must exist from the base of the tower to the exit door at the top of the tower. The player begins by ascending the spiral staircase gadgets from the starting area to the main platform gadget, which is the base of the tower's main segment. By traversing the staircase segment of the tower, the player will gain a height of $3\cdot(2n^2 + 2m)$. Since it takes at most 9 seconds for the player to traverse a spiral staircase gadget, the death plane will have risen by a height of $2n^2 + 2m$ in this time.

Once at the main platform, the player will set the truth value of each variable according to the previously stated satisfying assignment of the 3-SAT formula. Players can set the truth value of a variable by entering its variable gadget through one of its cracked-block pathways and removing the connecting block to one of the towers. This causes the tower to collapse. Since each tower corresponds to a particular truth assignment of a variable, the player should remove the connecting block from the tower which represents the undesired truth assignment. For example, if the player desires to set the variable to True, they would remove the connecting block from the False assignment tower.

The player can now use this removed block to help fill in the gap (recall figure \ref{fig:cathmaingap}) in the pathway to the clause gadgets. The player accomplishes this by pulling the removed block from the variable gadget through the unused cracked-block pathway and pushing it into the gap in the pathway to the clause gadgets. By lemma \ref{lem1}, this gap can be crossed by pushing $n$ blocks into it. The player can obtain these $n$ blocks by setting the truth assignment of all $n$ variables.

The truth assignment of each variable can be set using at most $20$ actions and the block removed from each variable gadget can be pushed into the gap using at most $6n$ actions. Thus, to set the truth assignment of every variable and traverse the gap, it would take at most $n\cdot(6n+20)$ actions. Recall, Vincent can conservatively execute one action per second, thus the traversal would take at most $n\cdot(6n+20)$ seconds. Since the death plane rises at a rate of 1-block per 9 seconds, it will be at a height of $2n^2 + 2m + \frac{n(6n+20)}{9}$, which is below the main platform's height of $3\cdot(2n^2 + 2m)$.

The player can traverse a clause by obtaining two blocks to push into the clause gap gadget (recall figure \ref{fig:cathclausegap} and figure \ref{fig:cathedge}). The player may use the extra block placed in the north-west corner to connect the platform to one of adjacent towers. Note, when this extra block is pushed off the platform, it will not fall since it is supported by the edge property. Consequently, the player now has access to one of the tower(s) and may extract a block from it. Notice, this tower does not collapse due to the edge property. The player can reach the next clause gadget by using the extracted block along with the extra block (which was previously used to connect to a tower) to fill the gap in the pathway to the next spiral staircase. Since a satisfying assignment exists for our embedded 3-SAT formula, each clause will have at least one adjacent tower that the player can pull a block from with the help of the extra block. Thus, every clause can be traversed. When the final clause is traversed, the player reaches the exit door. Each clause can be traversed using at most 20 actions (and consequently 20 seconds). Thus, to traverse all clauses, it would take the player at most $20m$ seconds. In this time, the death plane would have risen to a height of $2n^2 + 2m + \frac{n(n+20)}{9} + \frac{20m}{9}$, which is still below the main platform's height of $3\cdot(2n^2 + 2m)$. Consequently, given a satisfying assignment for the embedded 3-SAT, the player will always be able to reach the top of the tower before the death plane reaches the main segment of the tower.

Now suppose that a path exists from the starting area to the exit door, the 3-SAT that we embedded is satisfiable. For a path to exist, a player must enter each variable gadget and extract at exactly 1-block from one of its towers. The player cannot extract more than 1-block from a variable gadget since each of its connecting pathways can only be traversed twice before it collapses and both Vincent and the block count against a \textit{crack-block's} durability. Suppose, a player attempted to push two blocks over a cracked-block pathway. This would not be possible since each time you traverse a cracked block pathway it durability decreases by one. It is, however, possible to collapse both towers within the variable gadget, which is fine, and is effectively declaring that the variable will not be used. By lemma \ref{lem1}, the player cannot traverse the main platform gap gadget without $n$ blocks.  Thus, the player must take exactly 1-block from each of the $n$ variable gadgets in order to obtain the $n$ blocks necessary to fill in the gap in the pathway to the clause gadgets.

At each satisfied clause gadget, there must exist at least one tower adjacent to it. The player may extract a block from the tower by using the extra block placed in the north-west corner of the platform to form a bridge between the platform and the tower. Notice, the tower does not collapse since it is supported by the edge property. Without at least one tower adjacent to a clause gadget, the player will not be able to traverse the clause gadget, since they would not be able to obtain 2 blocks to place in the clause gap gadget to make it traversable. Suppose, the player tried to cross the gap without using two blocks. Since there is no jump mechanic, they would fall into the gap and be unable to get out of it because the gap has a height of 2 blocks. Now, suppose the player tried to just use the extra block placed on each platform to cross the gap. If they push this block into the gap, it will hover over the first slot of the gap due the edge property. The second slot of the gap would still be open and uncrossable. If the player jumped into the gap, they would be stuck as the in previous case. Moreover, the player would not be able to use blocks from one clause's tower to satisfy another clause since clauses are connected by spiral staircases and there is no mechanic in the game to push blocks up to another level.

If a tower does exist adjacent to a clause gadget, then the clause must be satisfiable since each adjacent tower represents a variable that evaluates to True for the clause. For a path to exist from the starting area to the exit door, all clause gadgets must be traversable. Therefore, each of them must have at least one adjacent tower to pull a block from. This necessarily means that an assignment of variables exists that makes all of the clause gadgets satisfiable. Since all clause gadgets are satisfied, our embedded 3-SAT is satisfiable.

\subsection{In NP}
Because every level has a time limit based on the position of the death plane, the plane moves up after a constant amount of time, and the plane must be placed somewhere on the map, so long as we assume the level itself has only polynomial height, this must mean that if there is a sequence of inputs to make it to the top, it can be verified in polynomial time by simply checking to see if the sequence allows Vincent to reach the top of the tower.

\bibliographystyle{plain}
\bibliography{Block_game_hardness}

\begin{thebibliography}{10}

\bibitem{aloupis2015classic}
Greg Aloupis, Erik~D Demaine, Alan Guo, and Giovanni Viglietta.
\newblock Classic nintendo games are (computationally) hard.
\newblock {\em Theoretical Computer Science}, 586:135--160, 2015.

\bibitem{FezInsp1}
amidos2006.
\newblock Trailer: ustwo’s monument valley.

\bibitem{demaine2020recursed}
Erik Demaine, Justin Kopinsky, and Jayson Lynch.
\newblock Recursed is not recursive: A jarring result.
\newblock {\em arXiv preprint arXiv:2002.05131}, 2020.

\bibitem{demaine2003tetris}
Erik~D Demaine, Susan Hohenberger, and David Liben-Nowell.
\newblock Tetris is hard, even to approximate.
\newblock In {\em International Computing and Combinatorics Conference}, pages
  351--363. Springer, 2003.

\bibitem{FezInsp3}
Jon Denton.
\newblock Secrets of raetikon review.

\bibitem{FezSales}
David Hing.
\newblock Fez hits 1 million sales.

\bibitem{FezInsp2}
Kristyna Hougaard.
\newblock Crossy road and unity ads.

\bibitem{schaefer1978complexity}
Thomas~J Schaefer.
\newblock On the complexity of some two-person perfect-information games.
\newblock {\em Journal of Computer and System Sciences}, 16(2):185--225, 1978.

\bibitem{shannon1950xxii}
Claude~E Shannon.
\newblock Xxii. programming a computer for playing chess.
\newblock {\em The London, Edinburgh, and Dublin Philosophical Magazine and
  Journal of Science}, 41(314):256--275, 1950.

\bibitem{BabaGDCTalk}
Arvi Teikari.
\newblock Reading the rules of baba is you.

\bibitem{wik}
{Wikipedia contributors}.
\newblock Baba is you --- {W}ikipedia{,} the free encyclopedia, 2021.
\newblock [Online; accessed 12-April-2020].

\bibitem{wikCath}
{Wikipedia contributors}.
\newblock Catherine (video game) --- {W}ikipedia{,} the free encyclopedia,
  2021.
\newblock [Online; accessed 11-April-2021].

\end{thebibliography}

\end{document}